\UseRawInputEncoding

\documentclass[journal]{IEEEtran}
\usepackage{cite}
\usepackage{graphicx,amssymb,lineno}
\usepackage{amsmath,amsfonts,amssymb,amsthm}
\newtheorem{theorem}{Theorem}

\newtheorem{definition}{Definition}
\usepackage{algorithm}
\usepackage{algorithmic}
\usepackage[usenames]{color}
\usepackage{float}
\usepackage{mathtools}
\usepackage{cite}
\usepackage{amstext}
\usepackage{color}

\usepackage{graphicx,graphics,color,epsfig,subfigure,graphpap,rotate}
\usepackage{times, verbatim, subfigure, epsfig, graphicx, latexsym, amsmath}
\usepackage{url}
\usepackage{subfigure}
\usepackage{booktabs}

\DeclareMathOperator*{\argmax}{arg\,max}

\begin{document}

\title{\huge{Coalition Game based User Association for mmWave Mobile Relay Systems in Rail Traffic Scenarios}}

\author{Chen~Chen,
        Yong~Niu,~\IEEEmembership{Member,~IEEE},
        Shiwen~Mao,~\IEEEmembership{Fellow,~IEEE},
        Xiaodan~Zhang,
        Zhu~Han,~\IEEEmembership{Fellow,~IEEE},
        Bo~Ai,~\IEEEmembership{Senior Member,~IEEE},
        Meilin~Gao,~\IEEEmembership{Member,~IEEE},
        Huahua~Xiao,
        and Ning~Wang,~\IEEEmembership{Member,~IEEE}

\thanks{Copyright (c) 2015 IEEE. Personal use of this material is permitted. However, permission to use this material for any other purposes must be obtained from the IEEE by sending a request to pubs-permissions@ieee.org.
This study was supported by National Key R\&D Program of China (2020YFB1806903); in part by the National Natural Science Foundation of China Grants 61801016, 61725101, 61961130391, and U1834210; in part by the State Key Laboratory of Rail Traffic Control and Safety (Contract No. RCS2021ZT009), Beijing Jiaotong University; and supported by the open research fund of National Mobile Communications Research Laboratory, Southeast University (No. 2021D09); in part by the Fundamental Research Funds for the Central Universities, China, under grant number 2020JBZD005; and supported by Frontiers Science Center for Smart High-speed Railway System; in part by the State Key Laboratory of Rail Traffic Control and Safety, Beijing Jiaotong University, under Grant RCS2019ZZ005; in part by the Fundamental Research Funds for the Central Universities 2020JBM089; in part by the Project of China Shenhua under Grant (GJNY-20-01-1); in part by ZTE Corporation, and State Key Laboratory of Mobile Network and Mobile Multimedia Technology; in part by US NSF CNS-2128368, CNS-2107216 and Toyota; in part by the NSF under Grant ECCS-1923717. (\emph{Corresponding author: Y. Niu.})}

\thanks{C. Chen is with the State Key Laboratory of Rail Traffic Control and Safety, Beijing Jiaotong University, Beijing 100044, China, and also with the Frontiers Science Center for Smart High-speed Railway System, Beijing Jiaotong University, Beijing 100044, China (e-mail: 20120026@bjtu.edu.cn).}

\thanks{Y. Niu is with the State Key Laboratory of Rail Traffic Control and Safety, Beijing Jiaotong University, Beijing 100044, China, and also with the National Mobile Communications Research Laboratory, Southeast University, Nanjing 211189, China (e-mail: niuy11@163.com).}

\thanks{S. Mao is with the Department of Electrical and Computer Engineering, Auburn University, Auburn, AL 36849-5201 USA (e-mail: smao@ieee.org).}

\thanks{X. Zhang is with the School of Management, Shenzhen Institute of Information Technology, Shenzhen 518172, China (e-mail: zhangxd@sziit.edu.cn).}

\thanks{Z. Han is with the Department of Electrical and Computer Engineering in the University of Houston, Houston, TX 77004 USA, and also with the Department of Computer Science and Engineering, Kyung Hee University, Seoul, South Korea, 446-701. (e-mail: zhan2@uh.edu).}

\thanks{B. Ai is with the State Key Laboratory of Rail Traffic Control and Safety, Beijing Jiaotong University, Beijing 100044, China, and also with the Beijing Engineering Research Center of High-speed Railway Broadband Mobile Communications,Beijing Jiaotong University, Beijing 100044, China (e-mail: aibo@ieee.org).}

\thanks{M. Gao is with the Tsinghua Space Center, Tsinghua University, Beijing 100091, China (e-mail: gaomeilin@mail.tsinghua.edu.cn).}

\thanks{H. Xiao is with the ZTE Corporation and State Key Laboratory of Mobile Network and Mobile Multimedia Technology, Shenzhen 518057, China (e-mail: xiao.huahua@zte.com.cn).}

\thanks{N. Wang is with the School of Information Engineering, Zhengzhou University, Zhengzhou 450001, China (e-mail: ienwang@zzu.edu.cn).}
}

\maketitle

\begin{abstract}
Rail transportation, especially, high-speed rails (HSR), is an important infrastructure for the development of national economy and the promotion of passenger experience. Due to the large bandwidth, millimeter wave (mmWave) communication is regarded as a promising technology to meet the demand of high data rates. However, since mmWave communication has the characteristic of high attenuation, mobile relay (MR) is considered in this paper. Also, full-duplex (FD) communications have been proposed to improve the spectral efficiency. However, because of the high speed, as well as the problem of penetration loss, passengers on the train have a poor quality of service. Consequently, an effective user association scheme for HSR in mmWave band is necessary. In this paper, we investigate the user association optimization problem in mmWave mobile relay systems where the MRs operate in the FD mode. To maximize the system capacity, we propose a cooperative user association approach based on coalition formation game, and develop a coalition formation algorithm to solve the challenging NP-hard
problem. We also prove the convergence and Nash-stable property of the proposed algorithm. Extensive simulations are done to show the system performance of the proposed scheme under various network settings. It is demonstrated that the proposed distributed low complexity scheme achieves a near-optimal performance and outperforms two baseline schemes in terms of average system throughput.
\end{abstract}


\section{Introduction}\label{S1}

As an important part of rail transportation, high speed rail (HSR) systems are able to provide reliable and convenient services to travelers. HSRs are booming and deployed all over the world. The total mileage of HSR will reach 54,550 km by 2025. As a result, it is extremely important to provide high quality wireless services for passengers. Since the demand of bandwidth resource is rapidly increasing, the current 4G network can no longer provide fast and stable services for passengers to enjoy Internet services onboard. Consequently, it is urgent to bridge the gap between offering high data rates for users and the existing low-rate communication networks for HSRs.

As a technology adopted by the fifth-generation (5G) wireless communications and beyond, millimeter wave (mmWave) communication can support multi-gigabit wireless services such as online gaming, video calling and so on~\cite{J1}. Thus, mmWave communications provides an effective solution to the problem of train-to-ground wireless broadband transmissions to enable stable and reliable communication services for travelers. Now, in some HSR communication systems worldwide, the technology has been applied and achieved superior performance, such as the Japanese Shinkansen. As for its wide bandwidth and narrow beamwidth, adopting mmWave technology for HSRs not only is well aligned with the current development trends of wireless communication networks, but has great potential to satisfy the demands of travelers for broadband wireless access. It has been recognized as an efficient means to enhance the quality of HSR services.

In Fig.~\ref{fig1}, we illustrate a HSR scenario where a mmWave mobile relay system is deployed. Since mmWave communications have short wavelength and suffer from high propagation loss. So the BS operating in the mmWave mode has the problem of small coverage. Thus, in this paper, we adopt the directional antennas that are leveraged to expand the rail coverage by utilizing the beamforming technology. Also, as the train's body is made of metal material, there is severe penetration loss to receiving signal power. However, by introducing MRs, a two-hop communication link is formed, which can effectively avoid the adverse effects of penetration loss and small coverage. In addition, advanced antenna technology and effective algorithm can be used to reduce the influence of Doppler shift and frequent handover, so as to ensure the effectiveness and reliability of train-to-ground communication.

Relays have been considered in the 3GPP Release 10 standard. In 3GPP Release 11, mobile relay (MR) support is added, which provides a theoretical basis for adoption of relays in high-speed trains. Compared with a direct channel, the two-hop relay link can achieve a high system throughput~\cite{a1}. Thus, mobile relay systems are able to provide stable and multi-gigabit services for HSR passengers. To further enhance the system performance, full-duplex (FD) technology has got great concern in the research field.
In theory, the spectral efficiency can be doubled by the the FD technology comparing with the half-duplex (HD) technology. However, due to the leakage of the transmitted signal to the receiver circuit that is called self-interference (SI), full-duplex communication systems were considered as impracticable before~\cite{FD1}. With the SI cancellation technology developing, such as antenna separation~\cite{SIC1}, beamforming-based techniques~\cite{SIC2}, and balanced/unbalanced (BALUN) cancellation~\cite{FD}, FD is now regarded as feasible. However, the residual self-interference (RSI) is still inevitable and SI cannot be completely eliminated in practice. Integrating FD technology into the MRs makes it efficient and feasible to improve the efficiency of bandwidth resource usage and enhance the system performance in terms of system sum rate. Therefore, MRs operate in the FD mode in our investigated system. There are one or more MRs and one base station (BS), and MRs are connected to the backbone network via the BS. Since mobile users can choose to be associated with the BS or MRs, user association problem is a key component of the mmWave MR system in HSR scenarios.

In the high-speed train, the users are relatively concentrated and have the characteristics of the same speed and direction of movement. In addition, different from users in traditional communication network, users in HSR communication have the characteristics of relativity and absoluteness of motion, that is, users are stationary relative to the train compartment, but they move at high speed relative to the ground. Therefore, the MR can be fully used to centrally control and manage the wireless access and resource allocation of the users in the high-speed train compartment and reduce the signaling cost.

However, with the emergence of a large number of new services and the continuous expansion of mobile application scenarios, the contradiction between the ubiquitous high-speed access demand of users and the complex wireless environment is gradually highlighted, which also brings huge challenges. Since the users are accustomed to high quality service on the ground, they also need to access the Internet during a long travel, so as to work, study and entertainment. Nevertheless, as the train runs at high speed, frequent handovers occur in the HSR scenarios. Assuming that the speed of the train is 350km/h and the radius of the cell is 1 to 2km, handover occurs every 10 to 20s~\cite{ho}. In addition, in the HSR communication system, the high-speed movement of the train relative to the BS will lead to more serious doppler shift and the fast time-varying fading of the channel. The above mentioned three factors caused by high speed will seriously affect the overall quality of service (Qos) of the HSR communication system. Fortunately, by leveraging the continuous location information of the trains, the Doppler shift can be pre-compensated according to~\cite{de}. Moreover, In traditional networks, user association is normally achieved based on the Reference Signal Received Power (RSRP) or Signal to Interference-plus-Noise Ratio (SINR), which are not feasible in HSR communication systems. In this case, users usually has tendency to associate with the BS, which may cause the BS associating with the majority of users and the resource allocated to each user is small. Thus, the limited resources per user from the congested BS will lead to a low system sum rate, while the resources at mobile relays can be under-utilized. In order to improve the mmWave MR system performance, users need to be actively associated with the under-loaded MRs that can provide better services and reduce the burden on the BS.
To sum up, with the improvement of HSR speed and the diversification of communication services, the requirements for communication services in HSR environment are getting higher and higher. How to provide high quality and efficient communication service for high-speed railway users has become a key problem to be solved urgently. As a result, it is necessary to study on optimizing the user association in rail traffic scenario.

In this study, we concentrate on optimizing user association in the mmWave MR system where the MRs working in the FD mode. In the user association problem, as the BS and each MR occupies different band and users associated with the BS or the same MR adopt time division multiple access (TDMA) for transmissions, there is only RSI from the MRs. As a result, the more the users associated with the BS or the same MR, the less the bandwidth resource is allocated to each user, and the transmission rate will decrease. Thus, the interaction among users seriously influences the system performance. To this end, we find the coalition formation game theory,
which can make the players (i.e., the users) try their best to cooperate to achieve the maximum system throughput~\cite{W1,W2}. Thus, a novel user association optimization scheme based on the coalition formation game is proposed in this paper. Unlike the other existing coalition formation game algorithm based user association, we use a key preference order, called the utilitarian order, which aims to maximize the total utility, while the other user association algorithms focus on maximizing the individual payoffs based on the Pareto preference order. The main contributions of this paper are summarized as follows.

\begin{itemize}
\item We formulate the problem of user association in the mmWave MR system as a nonlinear 0-1 programming problem. Users can choose to associate with the BS or a MR. We aim to optimize the user association problem in terms of average system throughput.

\item We propose an efficient coalition game based algorithm with low computational complexity to maximize the average system throughput. Then, we prove the properties of the proposed algorithm, such as convergence. We also give a proof that the proposed algorithm is Nash-stable.

\item In the section of performance evaluation, we do simulations with kinds of parameters and demonstrate the proposed coalition game based scheme outperforms two baseline user association strategies, and achieves a performance that is close to the optimal solution. Moreover, we do simulations on conditions of FD and HD, and provide insights on the impact of several essential network settings on the system performance.
\end{itemize}

The following describes the organization of the rest paper. Section~\ref{S2} gives an overview of related work. In Section~\ref{S3}, we introduce the system model and formulate the user association optimization problem. The coalition game model, algorithm and property analysis are introduced in Section~\ref{S4}. Section~\ref{S5} presents our performance evaluation. Finally,
Section~\ref{S6} gives the conclusion of this paper.

\section{Related Work}\label{S2}

There has been considerable work on mmWave communications in railroad transportation systems, especially for HSRs. In~\cite{J23}, the authors minimized the communication interruption time during handover on the basis of a distributed antenna system-based (DAS) mmWave communication system that is proposed for HSRs. To improve the area of rail coverage, Liu \emph{et al.}~\cite{a5} proposed a method based beamforming, and a semidefinite relaxation (SDR) method is developed to solve the optimization problem. For Long Term Evolution for Railways(LTE-R), a multiple-input multiple-output (MIMO) channel model was adopted~\cite{a6}.
Moreover, in~\cite{a7}, the authors developed a power allocation optimization algorithm of train-ground mmWave communication for HSRs to achieve the goal of energy-efficient. In~\cite{J9}, a fast initial access (IA) scheme was proposed for dual-band HSRs wireless networks. The authors exploit the historical IA results and regularity of train mobility to enhance the IA process of HSR communications. In~\cite{a4}, a blind beam alignment scheme based on deep reinforcement learning was presented to provide high rate services. In~\cite{a2}, sub-6-GHz bands and mm-wave bands were integrated as a network architecture to improve the area of network coverage while increasing throughput. In~\cite{J7}, a disaster radar detection approach was presented for train safety. In~\cite{J8}, the authors demonstrated that both orthogonal frequency-division multiplexing (OFDM) and the single carrier (SC) technology should be deployed to support train-trackside mmWave systems.

The research of user association has been done in various wireless networks. In~\cite{a8}, the authors provided a survey of user association mechanisms in mmWave communications. In HetNets, to maximize the transmission rate of the downlink, a dynamic user association was proposed for in~\cite{J10}. The authors used convex optimization method to derive an upper bound on the downlink, and then introduced a heuristic user association rule, which could approach the upper bound mentioned before and was also simple. In~\cite{J11}, the rate maximization problem under multi-cell and multi-user QoS constraints was formulated as a mixed integer nonlinear programming problem, and then feasible infeasible-interior-pointmethod (IIPM) was applied to optimize the system throughput. In~\cite{J12}, the authors jointly optimize the problem of user association and resource allocation in wireless HetNets, which is aimed to achieve high system capacity. A centralized iterative algorithm was proposed for collaborative optimization and joint resource allocation of user association, channel allocation, beam forming, and power control. In~\cite{J13}, the authors proposed a user-base station association scheme with coordinate descent under the proportional fair goal. This distributed scheme jointly considered beamforming, power control, and user association. In~\cite{J14}, by using the method of the convex optimization, the primal deterministic user association was relaxed to a new fractional association.  In~\cite{J15}, in order to minimize the delay of service requested by the users in femtocell networks, approximation algorithms were introduced to achieve a proven performance bound by formulating user association as a combinatorial problem.

Given the moving direction and route of HSR and users who are in the carriage are deterministic, and the moving track and position are regular and predictable, artificial intelligence (AI) such as machine learning (ML) is used to handle the user association problem, where users and mmWave BSs are treated as samples and classes~\cite{ML}. Simulation results showed that this proposed scheme had a superior performance while with low complexity. Moreover, by taking mobility into consideration, authors in~\cite{MB} proposed an user association strategy to overcome frequent handovers which are caused by the high speed of the trains. However, all these prior works did not integrate FD technology with MRs. In addition, given the high-mobility of HSRs, the Doppler effect is an important factor that can not be ignored. However, as the wireless channel is mainly line-of-sight(LOS) in HSR mmWave communications systems, there is only the Doppler shift~\cite{DS}. And by leveraging the continuous location information of the trains, authors in~\cite{de} proposed a scheme to pre-compensate the Doppler shift. Therefore, the Doppler effect was not taken into consideration in this paper. To sum up, although the prior works have made considerable advances, the existing user association schemes may not be effective for HSR wireless communication systems, which supports a variety of applications, some of which have stringent delay requirements (e.g., train safety related applications).

Game theory has been highly applied in the user association optimization problem. In~\cite{J16}, a non-cooperative game was introduced and the utility was the system sum rate of all the users. In~\cite{J17}, the authors introduced the coalitional games
to address wireless communications and networking problems. In~\cite{J18}, a bargaining problem based user association was proposed to maximize system utility, and the authors also proposed threshold minimum rate for the users to maintain fairness for all users.
By extending~\cite{J18}, in~\cite{J19}, human-to-human communication was labeled as first order business and machine-to-machine communication as second order business according to the proposed opportunistic user association algorithm. Furthermore, the proposed algorithm focused on allocating resource fairly for machine-to-machine traffic without lowering the QoS of human-to-human traffic. Moreover, a many-to-one matching game was proposed to solve the downlink user association optimization problem in HetNets~\cite{J20,J21,J22}, where users and BSs matched each other on basis of the predetermined principles. In~\cite{J20}, on basis of users utility functions, which accounted for not only the system throughput but also the fairness to cell edge users, BSs and users ranked each other. In~\cite{J21}, the handover failure probability and heterogeneous QoS requirements of users were mainly considered while designing utility function based user association. On basis of the problem formulated in~\cite{J21},
the user's quality of experience (QoE) in terms of mean opinion scores (MOS) and multimedia data services ware taken into consideration in ~\cite{J22}.

As much as we know, these prior works are all on basis of half-duplex communications, and do not consider the challenges in
mmWave mobile relay systems. So there is considerable space for performance improvement in terms of average system throughput.
In this paper, we study the user association optimization problem in mmWave mobile relay system with FD communications, and propose an algorithm based on coalition formation game to achieve excellent performance on average system throughput.

\section{System Overview and Problem Formulation} \label{S3}

In this section, we introduce the system model in~\ref{S3-1} firstly. Then, an example is illustrated in~\ref{S3-2} to make the proposed system model more clear. Finally, we formulate the
user association optimization problem in~\ref{S3-3}.

\subsection{System Model}\label{S3-1}

In this study, we consider a train-ground communication network with mobile relays operating in the FD mode for downlink transmissions, as shown in Fig.~\ref{fig1}. There is a BS on the roadside, and multiple MRs mounted on the train. The MRs are connected to the backbone network via the BS. The system serves both users outside the train and onboard the train. Part of the bandwidth resource is allocated to the BS for user association and communication with the MRs, and the rest is equally allocated to the MRs for user association. Moreover, the proportion of bandwidth resource allocation is fixed. We consider an FD-enabled mmWave mobile relay system. Considering the specifics of mmWave communications system where MRs operating in the full-duplex mode, each MR has two directional antennas, one is for transmitting and the other for receiving. The notation used in this paper is summarized in Table~\ref{tab:1}.

\begin{figure}[!t]
	\begin{center}
		\includegraphics[scale=0.4]{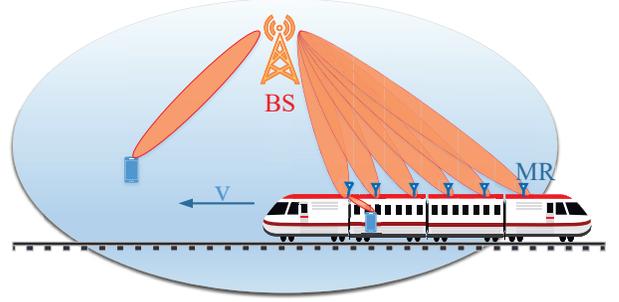}
	\end{center}
	\caption{Illustration of user association in the
		mmWave mobile relay system.} \label{fig1}
\end{figure}

In our system model, given the specifics of HSR mmWave communications systems, where the wireless channel is mainly LOS,
we use the Nakagami channel model. The channel coefficient corresponding to the BS and user $l$ is denoted as $h_{b,l}$ (i.e., the small-scale fading that obeys Nakagami-$m_i$ distribution with parameter $m_i$ as the fading depth), given by $h_{b,l}={|h_{0}|^{2}\cdot{d}^{- \delta}}$, where $|h_{0}|^{2}$ is fading power of the channel, $d$ is the user-BS distance, and $\delta$ is the path loss component. Let $G_{t}(b,l)$ and $G_{r}(b,l)$ be the transmit and receive antenna gains, respectively, and $P_{b}$ the transmit power of the BS. For user $l$ associated with the BS, the received power is given by:
\begin{equation}
P_{r}(b,l)=k_{0}|h_{0}|^{2}G_{t}(b,l)G_{r}(b,l){d}^{- \delta}P_{b}, \label{eq1}
\end{equation}
where $k_{0}$ is a constant coefficient proportional to ${(\frac{\lambda}{4\pi})}^2$ with $\lambda$ being the  wavelength. Further, thermal noise is modeled as
\begin{equation}
P_{noise}=N_{0}{\alpha_{i}W}, \label{eq2}
\end{equation}
where $N_{0}$ is the one-sided power spectral density of White Gaussian noise, $\alpha_{i}$ is the proportion of bandwidth allocation to the BS and MRs, and $W$ is the total channel bandwidth. Note that the $n$ MRs are indexed by $i=0,1,...n-1$ and the BS is indexed by $n$. In this paper, the spectrum allocation $\alpha_{i}$ is fixed and we will jointly optimize spectrum allocation and user association in our future works.

\begin{table}[!t]
\begin{center}
\caption{Notation}
\vspace{-0.05in}
\begin{tabular}{ll}
\toprule
Symbol & Description \\
\midrule
$n$ & number of MRs  \\
$A$ & number of users associated with the BS \\
$\textbf{A}$ & the set of users associated with the BS\\
$D$ & number of users associated with the MRs \\
$\textbf{D}$ & the set of users associated with the MRs \\
$x_{l,i}$ & whether the BS or MRs $i$ associated with user $l$  \\
$P_{r}(b,l)$  & received power at user $l$ from BS \\
$P_{r}(i,l)$  & received power at user $l$ from $MR_{i}$ \\
$G_{t}(b,l)$ & transmit antenna gain of BS associated with user $l$ \\
$G_{r}(b,l)$  & receive antenna gain of user $l$ associated with BS\\
$G_{t}(i,l)$ & transmit antenna gain of $MR_{i}$ associated with user $l$ \\
$G_{r}(i,l)$  & receive antenna gain of user $l$ associated with $MR_{i}$\\
${\beta}$ & self interference cancellation of MR \\
${\alpha_{i}}$ &  bandwidth allocation fraction to the BS or MRs $i$ \\
\bottomrule
\end{tabular}
\label{tab:1}
\end{center}
\end{table}

The BS and the MRs occupy different bands. Further, the users who are associated with the BS or the same MR adopt TDMA for transmission. Thus, there is no mutual interference for the BS-user transmissions, and the received signal to noise ratio (SNR) is given by
\begin{equation}
\mbox{SNR}=\frac{P_{r}(b,l)}{P_{noise}}, \label{eq3}
\end{equation}
Similarly, for user $l$ who is associated with MR $i$, the received power can be written as
\begin{equation}
P_{r}(i,l)=k_{0}|h_{0}|^{2}G_{t}(i,l)G_{r}(i,l){d}^{- \delta}P_{i}, \label{eq4}
\end{equation}
where $G_{t}(i,l)$ and $G_{r}(i,l)$ are transmit and receive antenna gain, respectively, $d$ is the distance from MR $i$ to the user $l$, and $P_{i}$ is the transmission power of MR $i$.

Since FD transmissions are used by the MRs, unlike the BS case, there is not only thermal noise, but also RSI. Specifically, $\mbox{RSI}=0$ indicates the most perfect and ideal self-interference cancellation case, and $\mbox{RSI}=1$ represents the absence of self-interference cancellation. The received SINR at user $l$ from MR $i$ is given by
\begin{equation}
\mbox{SINR}=\frac{P_{r}(i,l)}{P_{noise}+{\beta}P_{i}}. \label{eq5}
\end{equation}
where $\beta$ characterizes the SI cancellation ability of the MRs, $P_{i}$ is the transmit power of MR $i$. and the RSI can be denoted as ${\beta}P_{i}$. Without loss of generality, the transmit power for all the MRs is assumed to be same in this paper.

\subsection{An Illustrative Example}\label{S3-2}

In Fig.~\ref{fig2}, we provide an illustrative example of user association in the mmWave mobile relay system. There are two mobile relays plus one BS in this example, and thus, the total bandwidth resource $W$ is divided into three parts. We use rectangles of different colors to represent the three sub-bands. For each user, it is impractical to be associated with both the BS or the MRs.
As a result, we divide all the users into two categories: (i) the users outside the train associated with the BS and (ii) users on the train associated with the MRs. From Fig.~\ref{fig2}, we can see that users $u_{1}$, $u_{2}$, and $u_{3}$ are associated with MR1 and occupy the same spectrum band that is also used by MR1 to communicate with the BS (marked by yellow blocks). Similarly, $u_{4}$ and $u_{5}$ are associated with MR2 and occupy the other sub-band marked by green blocks, which is also used by MR2 to communicate with the BS. Finally, users $u_{6}$, $u_{7}$, $u_{8}$, and $u_{9}$ utilize the remaining bandwidth (marked in purple)
and are associated with the BS. Due to FD transmissions, there is RSI for those users associated with the MRs, which limits the performance of such users. On the other hand, mutual interference is negligible since the users who occupy the same sub-band adopt TDMA for transmissions.

\begin{figure}[!t]
	\begin{center}
		\includegraphics*[width=1\columnwidth,height=2.7in]{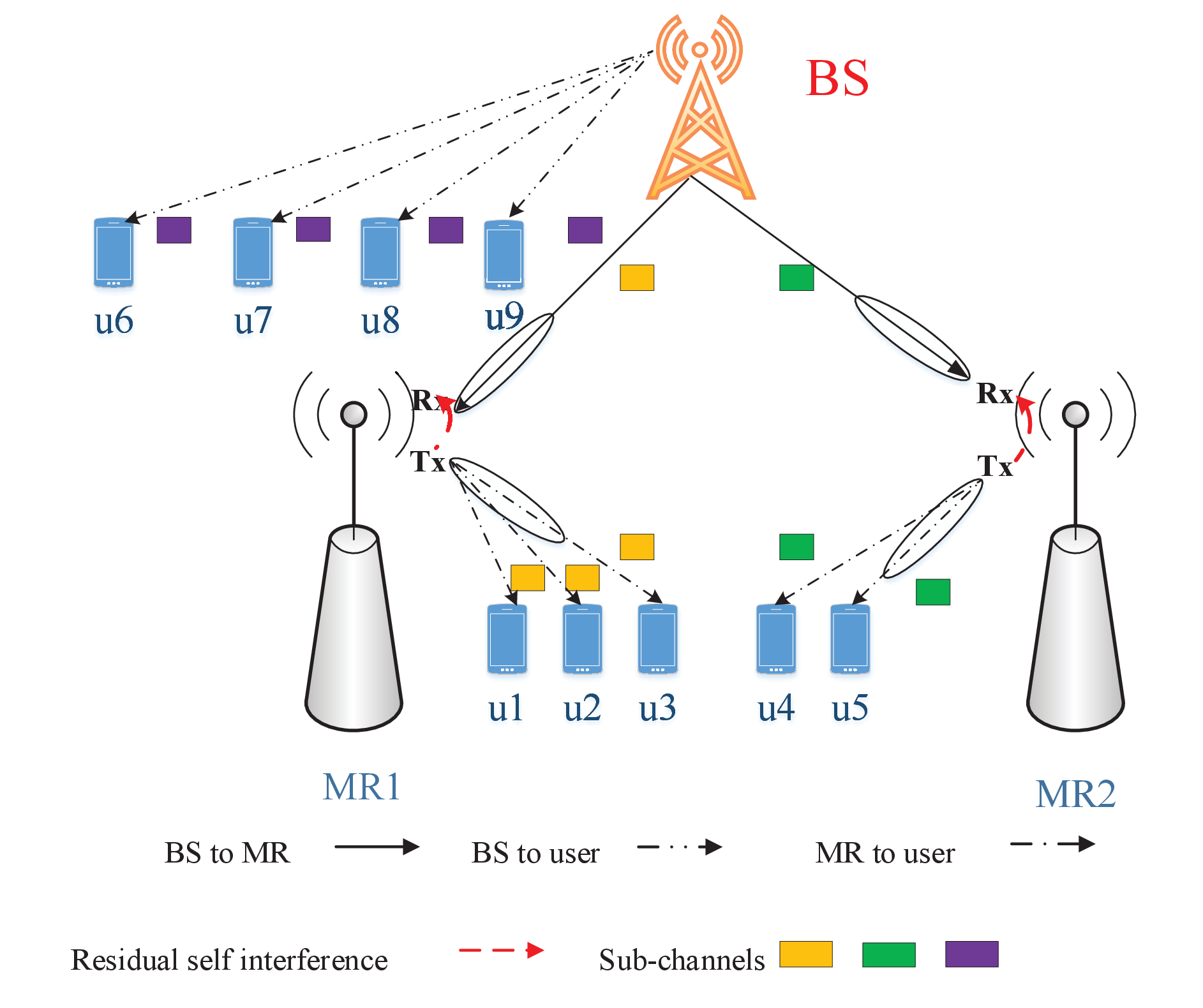}
	\end{center}
	\caption{An example of user association in the
             mmWave mobile relay system.} \label{fig2}
\end{figure}

\subsection{Problem Formulation}\label{S3-3}

We assume the proposed system has $n$ onboard MRs and one BS on the ground. There are $A$ users associated with the BS, denoted by $\textbf{A}=\{a_{1}, a_{2}, ..., a_{A}\}$. Moreover, we denote the set of users associated with the MRs by $\textbf{D}$, written as $\textbf{D}=\{d_{1}, d_{2}, ..., d_{D}\}$. The set of users associated with MR $i$ is denoted as $\textbf{D}_{i}$, $i=0, 1, ..., n-1$. We define binary variable for each user $l$, as $x_{l,i}$, $l \in \textbf{A} \cup \textbf{D}$ to indicate user association with the BS or MR $i$. If user $l$ is associated with the BS or MR $i$, we have $x_{l,i}=1$; otherwise, we have $x_{l,i}=0$.
The fraction of bandwidth allocated to an MR $i$ is denoted by ${\alpha_{i}}$, $i=0, 1, ..., n-1$, and the faction of bandwidth allocated to the BS is denoted by $\alpha_{n}$. For the users associated with the BS, from the Shannon's formula, the transmission rate of user $a$, denoted by $R_{a}$, is
\begin{equation}
R_{a}=x_{a,n}\alpha_{n}W{\log}_2\left({1+\mbox{SNR}_a}\right).
\label{eq6}
\end{equation}
For the users associated with MR $i$, there is RSI between the user and MR $i$. The transmission rate of user $d$ can be calculated as
\begin{equation}
R_{d}=x_{d,i}\alpha_{i}W{\log}_2\left({1+\mbox{SINR}_d}\right), i=0, 1, ..., n-1.
\label{eq7}
\end{equation}

The average transmission rate of the overall system can be expressed as
\begin{equation}
R=\frac { \sum_{a \in {\textbf{A}}}R_{a} + \sum_{i=0}^{n-1} \sum_{d\in{\textbf{D}_{i}}}R_{d} } {|\textbf{A}| + \sum_{i=0}^{n-1}|\textbf{D}_{i}| }.
\label{eq8}
\end{equation}

The system average throughput in~\eqref{eq8} depends on user association as indicated by $x_{l,i}$'s. We define a binary decision matrix $\mathbf{X}$ that composes of user association decision variables $x_{l,i}$, $l\in \textbf{A} \cup \textbf{D}$ and $i \in \{0, 1, ..., n-1, n\}$. And $R(\mathbf{X})$ denotes the average system throughput. It is the objective function that we aim
to optimize in this work. User $l\in \textbf{A} \cup \textbf{D}$ is associated with only one of the MRs or the BS for transmissions. In addition, the number of users that each MR and BS can serve is limited. Finally, the sum of $\alpha_{i}$'s should be 1. Thus, the constraints can be obtained as follows.
\begin{align}
& \sum_{i=0}^{n}x_{l,i}\leq1,\ \ \forall \ l\in \textbf{A} \cup \textbf{D}
\label{eq9} \\
& \sum_{i=n}x_{l,i}\leq Y, \ \text{if} \ l\in \textbf{A}
\label{eq10} \\
& \sum_{l\in \mathbf{D}_i} x_{l,i}\leq Z, \
\ i=0, 1, ..., n-1
\label{eq11} \\
& \sum_{i=0}^{n}\alpha_{i}=1,
\label{eq12}
\end{align}
where $Y$ is the maximum number of users that the BS can serve and $Z$ is the maximum number of users that each MR can serve.

Therefore, the user association problem of the FD mmWave mobile relay system is formulated as
\begin{align}
\max\,\,& R({\mathbf X}) \label{eq13} \\
\mbox{s.t.}\,\,&
(\ref{eq9}),~(\ref{eq10}),~(\ref{eq11}),\mbox{and}~(\ref{eq12}). \notag
\end{align}
This is a nonlinear integer programming problem with binary variables $x_{l,i}$. This is an NP-hard problem and more complex
than the 0-1 Knapsack problem~\cite{J24}. In the next section, we will propose an algorithm on basis of coalition game to solve the problem distributively with a low complexity. For each user in the system, the proposed algorithm makes association decision with the BS or the MR aiming to demonstrate a greater performance on the system utility function.

\section{Coalition Formation Game}\label{S4}

In this section, we will introduce the model of coalition formation game in~\ref{S4-1} firstly. Secondly, the algorithm based the coalition formation game is developed to make reasonable and effective user association decisions by cooperation in~\ref{S4-2}. Finally, we analyze the proposed algorithm and present its properties in~\ref{S4-3}.

\subsection{Coalition Formation Game: Model} \label{S4-1}

Coalition formation game is a subcategory of the cooperative game theory, which is able to determine how users cooperate to achieve the purpose of average system throughput optimization. We firstly introduce the definition of coalition formation game with transferable utility.

\begin{definition}
Coalition Formation Game with Transferable Utility: A coalition formation game with transferable utility is denoted by a pair ($C, U$). $C$ denotes the set of players and $U$ denotes the utility function. For all $C_{i}\subset C, \ U(C_{i})$ characterizes the obtained aggregate profit of coalition $C_{i}$.
\label{define1}
\end{definition}

We model user association optimization as a coalition formation game, where users have strong motivation to cooperate by forming coalitions, and further improve the network performance from the perspective of average system throughput. The set of users associated with the BS or the same MR consist a coalition. Consequently, the proposed game consist of $n+1$ coalitions and is denoted as $C=\{C_{0}, C_{1}, ..., C_{n-1}, C_{n}\}$, where $C_{i} \bigcap C_{i^{\prime}} = \emptyset$ for any $i\neq i^{\prime}$, and $\bigcup^{n}_{i=0}C_{i}=\textbf{A} \bigcup \textbf{D}$. Based on the Shannon's formulas in~(\ref{eq6}) and~(\ref{eq7}), the
average throughput of coalition $C_{i}$ can be written as
\begin{equation}
R(C_{i})=\frac {\sum_{a \in C_{i}}R_{a}+\!\!\sum_{d\in C_{i}}R_{d}}{|C_{i}|}. \label{eq14}
\end{equation}

According to the above definition, we further define the coalition formation game for~\eqref{eq13}.
\begin{definition}
Coalition Formation Game for User Association: The coalition formation game with transferable utility for user association
in the FD mmWave train-ground communication system is denoted as a pair ($C, U$) and the following items define the game.
\begin{itemize}
\item[$\bullet$] Players: \emph{a user either user associated with the BS $a \in \textbf{A}$ or a user associated with a MR $d \in \textbf{D}$.}
\item[$\bullet$] Coalition Partition: \emph{since there is one BS and n MR, $n+1$ coalitions are divided for the set of players
$C=\textbf{A} \bigcup \textbf{D}$, i.e., $C=\{C_{0}, C_{1}, ..., C_{n-1}, C_{n}\}$, where $C_{i}\bigcap C_{i^{\prime}}=\emptyset$ for $0\leq i, i^{\prime} \leq n$, $i, i^{\prime}\in \mathbb{N}$, and $i\neq i^{\prime}$.}
\item[$\bullet$] Transferable Utility:
\emph{for all $C_{i}\subset C$, $U(C_{i})$ is the total utility of the coalition $C_{i}$ that
comprises all the users associated with the same MR $i$ or the BS.
In this paper, we distribute utility to members on the basis of their individual contributions.}
\item[$\bullet$] Strategy:
\emph{on the basis of maximizing the total utility, each player
decides on whether to leave the current coalition and join another.}
\label{define2}
\end{itemize}
\end{definition}

In the above game formulation, the utility is the average system throughput, and users are inclined to form coalitions to maximize the utility of the system. The proposed coalition formation algorithm will be introduced in~\ref{S4-2}.

\subsection{Coalition Formation: Algorithm}\label{S4-2}

To achieve the purpose of maximizing average system throughput, each user needs to collaborate with other users to choose a
proper coalition to join. A preference order is needed for each user, by comparing and ordering its potential coalitions.
The preference order is set by coalition game to define different strategies for players in different scenarios. In general, under kinds of parameters, such as the approval of the other players in the same coalition and the payoff receiving from its coalition etc. We show how to obtain the preference order next.

For any player ($l \in \textbf{A} \bigcup \textbf{D}$), we use the notation ${\succ}_l$ to denote its preference order:
$C_{i} \ {\succ}_l \ C_{{i}^{\prime}}$, $i\neq {i}^{\prime}$ indicates that player $l$ prefers to join the coalition $C_{i}$ than coalition $C_{{i}^{\prime}}$. Then, according to \cite{J17}, the preference order for player $l\subseteq C_{i}, l\subseteq C_{{i}^{\prime}}$ is denoted by
\begin{equation}
C_{i} \ {\succ}_{l} \ C_{{i}^{\prime}} \Longleftrightarrow U(C_{i})+U(C_{{i}^{\prime}} \backslash l) > U(C_{i} \backslash l)+ U(C_{{i}^{\prime}}). \label{eq15}
\end{equation}

With the above definition, the players set preferences over all their candidate coalitions. On basis of the preference order,
each player chooses the coalition that can contribute to the total utility of the original and new coalitions. We define the switch operation for forming the final coalitions in the following.

\begin{definition}
	Switch Operation: For a given partition $C=\{C_{0}, C_{1}, ..., C_{n-1}, C_{n}\}$, for player $p\subset \textbf{A}\bigcup \textbf{D}$, when and only when the preference relation $C_{i} \ {\succ}_l \ C_{{i}^{\prime}}$, ${i} \neq {{i}^{\prime}}$ is satisfied, a switch operation moves player $l$ from $C_{{i}^{\prime}}$ to $C_{i}$. Then, $C$ is replaced by ${C}^{\prime}$,
as ${C}^{\prime}=(C\backslash \{C_{i}, C_{{i}^{\prime}}\})\bigcup \{{C_{{i}^{\prime}}\backslash \{l\}, C_{i}\bigcup \{l\}}\}$.
	\label{define3}
\end{definition}

The coalition formation algorithm for user association is shown in Algorithm~\ref{alg1}. To start with, we randomly form an initial partition. Then, on basis of this partition, iterations are done to achieve better performance. In Step 5, an user $l$ is selected
one by one from $i=0$ to $i=n$. Except the current coalition $C_{i}$, the user needs to select another coalition $C_{{i}^{\prime}}$. When user $l$ makes its decision, it checks not only whether the preference relation $C_{{i}^{\prime}} {\succ}_{l} C_{i}$ is satisfied, but also whether the number of users of $C_{{i}^{\prime}}$ is less than $Y-1$ or $Z-1$ in Step 6.
If both conditions are satisfied, the algorithm executes one switch operation and updates the partition in following Steps 8 and 9. Each player decides on whether to leave the current coalition when an operation of iterating occurs. As long as a final network coalition partition $C_{fin}$ has not obtained, the algorithm will keep iterating. We set specific termination condition for the coalition game algorithm. The parameter $j$ denotes the number of successive non-switching iterations. We set $j=10*(|\textbf{A}|+|\textbf{D}|)$. That is to say, for each player, there are ten chances for it to switch to another coalition.
If the switching operation has not been performed for ten times, the system is assumed to reach a stable state (see Definition~\ref{define4}) and an sub-optimal solution is obtained.

\begin{algorithm}[!t]
	\caption{The Coalition Formation Algorithm for User Association}
	\label{alg1}
	\begin{algorithmic}[1]
		\STATE Form the random initial
		partition
		$C_{ini}=\{C_{0}, C_{1}, ..., C_{n-1}, C_{n}\}$;
		\STATE Set the current partition as $C_{ini}\longrightarrow C_{cur}$ and $j=0$;
		\REPEAT
		\STATE Select one user
		$l\in \textbf{A} \bigcup \textbf{D}$
		in the pre-determined order, and its coalition is denoted as $C_{i}\subset C_{cur}$;
		\IF{$(C_{{i}^{\prime}} {\succ}_{l} C_{i})$ and $((|C_{{i}^{\prime}}|+1 \leq Y$,  ${i}^{\prime}=n)$ or $(|C_{{i}^{\prime}}|+1 \leq Z$, ${i}^{\prime}$=0, 1, ..., n-1))}
		\STATE $j=0$;
		\STATE Player $l$ leaves its current coalition $C_{i}$ and joins the new coalition $C_{{i}^{\prime}}$;
		\STATE Update the current partition
		as $(C_{cur} \backslash \{C_{{i}^{\prime}}, C_{i}\})\bigcup \{C_{i} \backslash \{p\}, C_{{i}^{\prime}}\bigcup \{p\}\}\longrightarrow C_{cur}$;
		\ELSE
		\STATE $j=j+1$;
		\ENDIF
		\UNTIL The final partition $C_{fin}$ reaches Nash-stable.
	\end{algorithmic}
\end{algorithm}

We next examine the complexity of Algorithm~\ref{alg1}. At first, the chosen player compares the total utility of the current coalition and the candidate coalition. Next, it decides on whether to leave the current coalition. Consequently, there is no more than one switch operation in each iteration. And if the total number of iterations is $I$, the complexity of the coalition formation algorithm is $\mathcal{O}(I)$.

\subsection{Coalition Formation:
Analysis}\label{S4-3}

\subsubsection{Convergence}\label{S4-3-1}

We first examine the convergence performance of the coalition formation algorithm.

\begin{theorem}
With an arbitrary initial partition $C_{ini}$ started, a final partition $C_{fin}$ of the proposed algorithm always be found. Through a finite number of switching operations, $C_{fin}$ consists of multiple disjoint coalitions.
\label{the1}
\end{theorem}

\begin{proof}
Users are divided into $U^{\prime}$ disjoint subsets in coalition structure (CS) and each subset denotes a coalition . We
exhaustively examine all the CSs of $U^{\prime}$ users and obtain a the set $S$ of all CSs, given by $S=\{S_{1}, S_{2}, ..., S_{|S|}\}$. Based on the Bell number $B_k$, the cardinality of $S$ for the $U^{\prime}$ users can be expressed as
\begin{equation}
B(U^{\prime})=\left\{
\begin{array}{lcl}
1,&& {N^{\prime}=0}, \\
\sum_{k=0}^{U^{\prime}-1}\binom{U^{\prime}-1}{k} B_{k},&& {U^{\prime} \geq 1}.\\
\end{array} \right. \label{eq16}
\end{equation}

In Algorithm~\ref{alg1}, the number of switch operations in each iteration is up to $|\textbf{A}|$ + $|\textbf{D}|$; each user
either joins a new coalition or still stays in the original coalition. Thus, the whole number of partitions is $B(|\textbf{A}|+|\textbf{D}|)$, and since we preset the number of the BS and MRs, the number of partitions will be decreased
based on the Bell number. As the number of partitions is bounded, we can conclude that after a certain number of switch operations, a final partition $C_{fin}$ will be obtained.
\end{proof}

\subsubsection{Stability}\label{S4-3-2}

We begin with the definition of Nash-stable structure .

\begin{definition}
Nash-stable Structure:
If for all $l \in \textbf{A}\bigcup \textbf{D}$, $l\subseteq C_{i}\subset C, C_{i} \succ_{l} C_{{i}^{\prime}} \bigcup \{l\}$, for all $C_{{i}^{\prime}} \subset C, C_{{i}^{\prime}}\neq C_{i}$, a coalitional partition $C=\{C_{0}, C_{1}, ..., C_{n-1}, C_{n}\}$ is Nash-stable.
\label{define4}
\end{definition}

Based on the definition given above, we propose the theorem~\ref{the2}

\begin{theorem}
The final partition $C_{fin}$ obtained by the proposed algorithm in Algorithm~\ref{alg1} is Nash-stable.
\label{the2}
\end{theorem}

\begin{proof}
The proposed algorithm based coalition formation game will terminate when a number of successive switch operations no longer occur. Meanwhile, the system throughput has maximum value and the final partition is Nash-stable. We have
${C_{i}}^{\ast} =\argmax_{C_{i}} U(C)$, for all $C_{i}\subset C$, and ${C}^{\ast}=\{{C}^{\ast}_{0}, {C}^{\ast}_{1}, ..., {C}^{\ast}_{n-1}, {C}^{\ast}_{n}\}$ is the final Nash-stable coalition partition.

If the obtained final partition $C_{fin}$ has not reached Nash-stable, i.e., there exists a player $p \in
\textbf{A}\bigcup\textbf{D}$, $C_{i} \succ_{p} C_{{i}^{\prime}} \bigcup \{p\}$ is not matched with all $C_{{i}^{\prime}} \subset C$, ${i}^{\prime} \neq {i}$. Therefore, the switch operation from $C_{i}$ to $C_{{i}^{\prime}}$ occurs and another final partition will be formed. Since the partition still has tendency of switching, the final partition $C_{fin}$ has not been found yet, which is contradictory to $C_{fin}$ is the final partition. Consequently, the final partition $C_{fin}$ of the proposed algorithm has reached Nash-stable.
\end{proof}

\begin{table}[!t]
\begin{center}
\caption{Simulation Parameters}
\vspace{-0.05in}
\begin{tabular}{lll}
\toprule
Parameter & Symbol  & Value \\
\midrule
System bandwidth & $W$ & 2,160 MHz  \\
BS transmission power & $P_{b}$ & 30 dBm \\
MR transmission power&  $P_{n}$ & 23 dBm \\
Noise spectral density & $N_{0}$ & -134 dBm/MHz \\
Path loss exponent  &  $\delta$   &  2 \\
Carrier wavelength  &   $\lambda$   & 5 mm \\
Half-power beamwidth  &   $\theta_{-3dB}$ &    $30^{\circ}$ \\
\bottomrule
\end{tabular}
\label{table2}
\end{center}
\end{table}

\section{Performance Evaluation}\label{S5}
\subsection{Simulation Setup}\label{S5-1}

\begin{figure*}[!t]
	\centering
	\subfigure[]{
		\begin{minipage}{3.2in}
			\centering
			\includegraphics[width=3.2in]{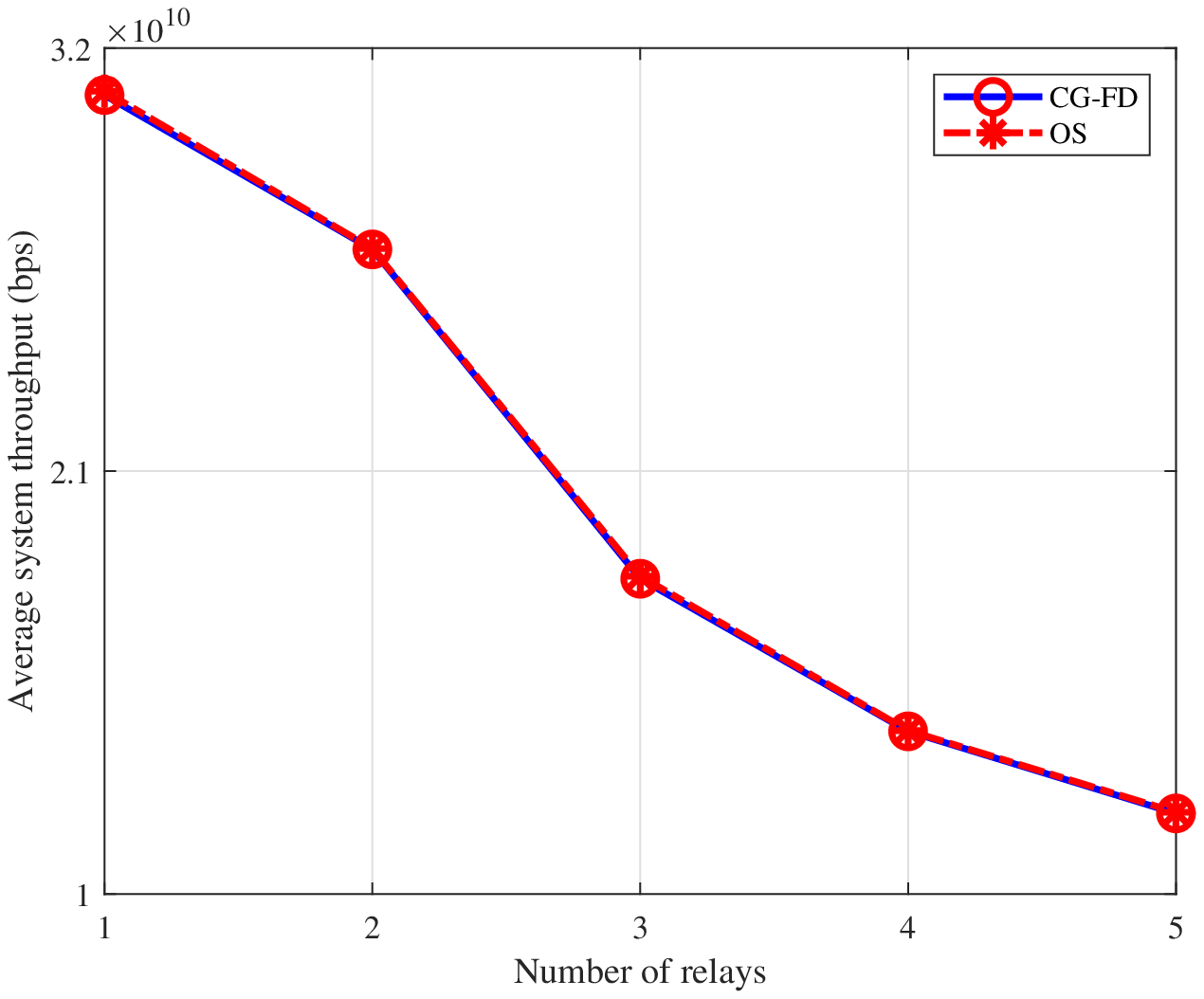}
		\end{minipage}
	}
	\subfigure[]{
		\begin{minipage}{3.2in}
			\centering
			\includegraphics[width=3.2in]{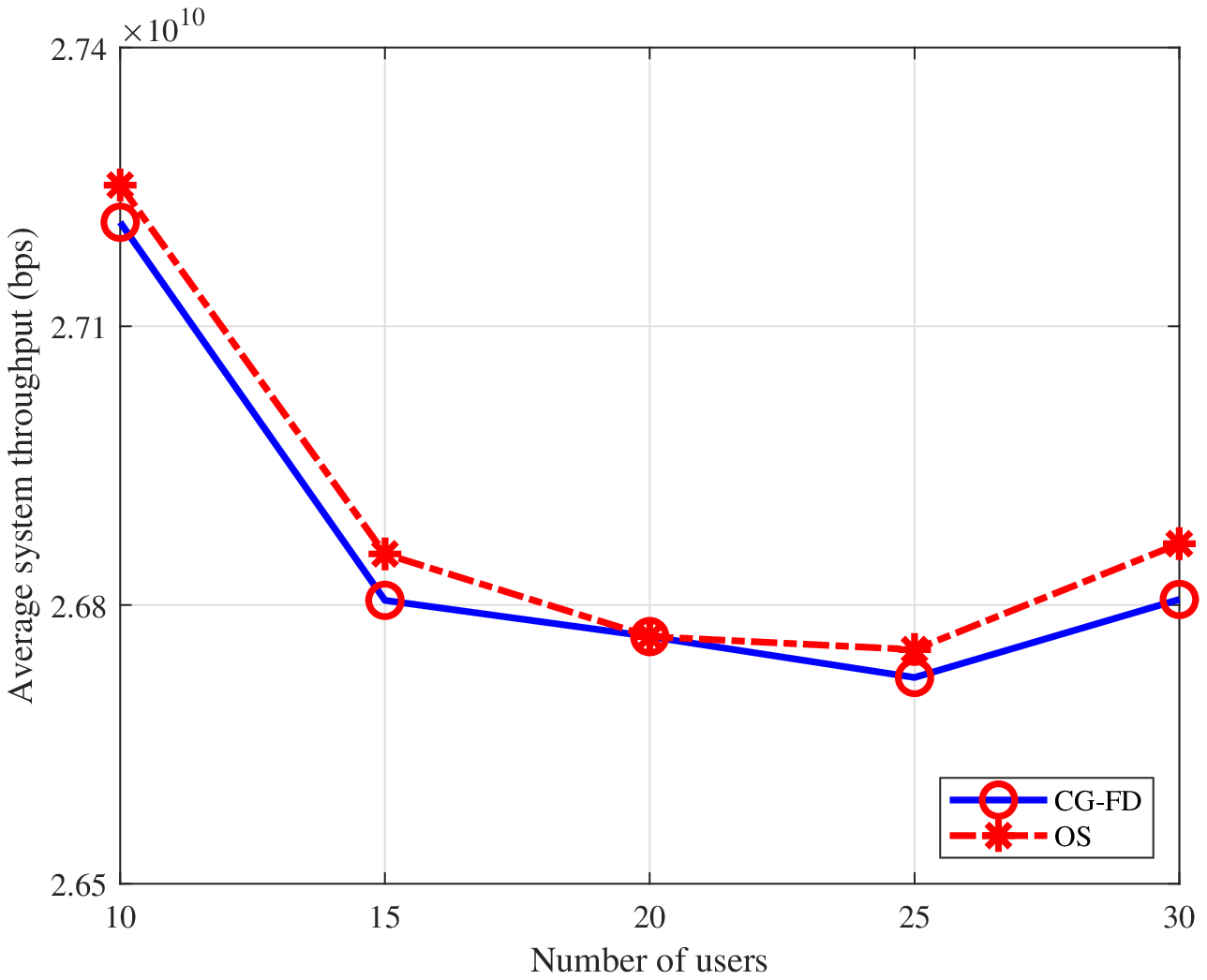}
		\end{minipage}
	}
	\caption{Average system throughput comparison of the proposed scheme CG-FD and the optimal solution (OS) under (a) different numbers of MRs; and (b) different numbers of users.}
\label{fig3}
\end{figure*}

In the simulations reported in this section, we use the directional antenna model in IEEE 802.15.3c for the mmWave communication system~\cite{J25}. On basis of the model, $G(\theta)$, the directional antenna gain denoted in units of decibel (dB), is given by
\begin{align}
	G({\theta})=\left\{
	\begin{array}{ll}
		G_{0}-3.01\cdot\left({\frac{2\theta}{\theta_{-3dB}}}\right)^2, & {0^{\circ}\leq {\theta} \leq {{\theta}_{ml}}/2} \\
		G_{sl}, & {{{\theta}_{ml}}/2 \leq {\theta} \leq 180^{\circ}}, \\
	\end{array} \right. \label{eq17}
\end{align}
where $\theta$ denotes an angle ranging from $0^{\circ}$ to $180^{\circ}$, ${\theta}_{ml}$ is the main lobe beamwidth, and ${\theta}_{-3dB}$ represents the angle of the half-power beamwidth.
Moreover, ${\theta}_{ml}=2.6\cdot {\theta}_{-3dB}$. $G_{0}$ denotes the maximum antenna gain, which can be expressed as
\begin{equation}
G_{0}=10\log\left(\frac{1.6162}{\sin(\theta_{-3dB}/2)}\right)^2.
\label{eq18}
\end{equation}
The side lobe gain $G_{sl}$ is given by
\begin{equation}
G_{sl}=-0.4111\cdot \ln(\theta_{-3dB})-10.579.
\label{eq19}
\end{equation}
The other relevant simulation parameters are listed in Table~\ref{table2}.

In the evaluation, we use the following two performance metrics to demonstrate the benefits of our proposedapproach.
\begin{itemize}
\item \textbf{Average system throughput}: The achieved average system throughput of all the users associated with the BS and MRs.

\item \textbf{Single user throughput}: The achieved throughput of single user associated with the BS or MRs.

\item \textbf{Number of switch operations}: Average number of switch operations to reach Nash-stable using the proposed algorithm.
\end{itemize}

For demonstrating the benefits of the proposed coalition game algorithm, we compare it with the following two baseline schemes. For simplicity, our proposed algorithm is termed \textbf{Coalition Game-Full Duplex} (CG-FD).
The two baseline schemes are:
\begin{itemize}
\item \textbf{Coalition Game-Half Duplex} (CG-HD), where the MRs operate in the half duplex mode. Thus, the bandwidth resource that the users can occupy is decreased in comparison with the case of FD under the same number of users.
\item \textbf{Non-cooperative Coalition Game-Full Duplex} (NCCG-FD), where the coalitions does not cooperate with each other i.e., the preference order defined in~(\ref{eq15}) is changed to $C_{i} \ {\succ}_{l} \ C_{{i}^{\prime}} \Longleftrightarrow U(C_{i}) > U(C_{{i}^{\prime}})$, which can not guarantee that as user $l$ leaves $C_{{i}^{\prime}}$ to join $C_{i}$, the average system throughput will increase.
\end{itemize}

\subsection{Comparison with the Optimal Solution}\label{S5-2}

To show superiority of our proposed scheme, we plot the system performance comparison of our proposed scheme and the optimal solution(OS) in terms of average system throughput under different users and MRs in Fig.~\ref{fig3}. Since the OS is achieved
through exhaustive search, it is of extreme complexity. Thus, the simulations are operated in the small number of users. In Fig.~\ref{fig3}(a), the number of users is to 20 and the number of MRs is varied from 1 to 5. Then, the number of MRs is set to 2 and the number of users is varied from 10 to 30 in Fig.~\ref{fig3}(b). We use average deviation of average system throughput
between our proposed scheme and optimal solution to do quantitative analysis and it is computed as follows.
\begin{equation}
\mbox{Average \ Deviation}= \frac{1}{5}\sum\limits_{p=1}^{5}{\frac{U_{OS}(p)-U_{CG-FD}(p)}{U_{OS}(p)}}, \label{eq26}
\end{equation}
where $U_{OS}(p)$ and $U_{CG-FD}(p)$ are the average system throughput of OS and CG-FD. From Fig.~\ref{fig3}(a) and Fig.~\ref{fig3}(b), it is obvious that our proposed scheme performs little worse than the OS and the result of the two schemes is very close.

\subsection{Comparison with Baseline Schemes}\label{S5-3}
\subsubsection{Impact of $\beta$}

Recall that ${\beta}$ denotes the SI cancellation capability. We plot the system performance of the three schemes in terms of average system throughput under different ${\beta}$. The number of users is set to 40 and the number of MRs to 2. The abscissa $x$ is the magnitude of ${\beta}$. From Fig.~\ref{fig4}, it is obvious that the CG-FD scheme performs best among the three schemes. In addition, when ${\beta}$ varies, the result of the CG-HD scheme does not change, since the half duplex algorithm has nothing to do
with SI cancellation. For the CG-FD scheme and the NCCG-FD scheme, the SI cancellation level influences their performance significantly. When ${\beta}$ is large, which indicates that the self-interference cancellation ability is weak, the performance of CG-FD and NCCG-FD are both worse than CG-HD. When ${\beta}$ becomes smaller, the RSI influences the transmission rate little. As a result, the throughputs of the two FD schemes become larger. When the ${\beta}$ decreases beyond $10^{-13}$, the average system throughputs of the two FD schemes do not increase further. At this level of ${\beta}$, the RSI is effectively suppressed and it has negligible impact on the transmission rate. In the remaining simulations in this section, we set ${\beta}$ to $10^{-13}$, which
ensures good performance for the two FD schemes.

\begin{figure}[!t]
\begin{center}
\includegraphics*[width=3.2in]{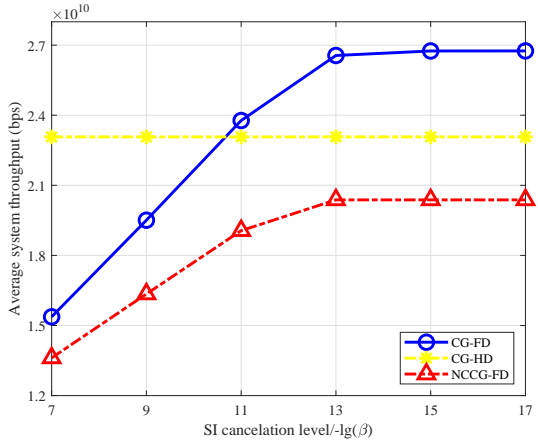}
\end{center}
\caption{Average system throughput of the three user association schemes with different ${\beta}$ values.}
\label{fig4}
\end{figure}

\subsubsection{Impact of Number of MRs}

In Fig.~\ref{fig5}, we plot the system performance of the three schemes in terms of average system throughput under the number of users to be 40 and changing the number of MRs from 1 to 6. In Fig.~\ref{fig5}, it is obvious that the CG-FD scheme has the best performance than the other two schemes. With more MRs, the average system throughput of all the schemes decrease. Since when the number of MRs is increased, there will be less bandwidth allocated to each MR. The transmission rate achieved by users who are associated with the MRs will be smaller and the average system throughput of the algorithms CG-FD, CG-HD and NCCG-FD will thus
decrease.

\begin{figure}[!t]
\begin{center}
\includegraphics*[width=3.2in]{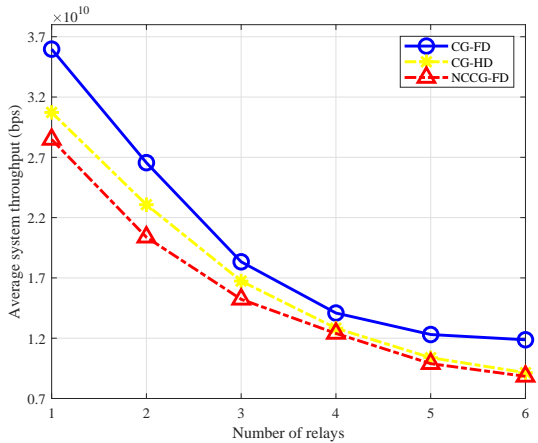}
\end{center}
\caption{Average system throughput of the three user association schemes under different numbers of MRs.}
\label{fig5}
\end{figure}

\subsubsection{Impact of Number of Users}

In Fig.~\ref{fig6}, we plot the system performance of the three schemes in terms of average system throughput under the number of MRs is 4 and the number of users is varied from 25 to 50. As users increasing, the average system throughputs of the three schemes
remain nearly constant. Moreover, the CG-FD scheme always has the best system performance in terms of average system throughput than the other two baseline schemes. In comparison with the CG-HD scheme, the CG-FD scheme allows the MRs to simultaneously communicate with users and the BS to achieve a higher system throughput, while the CG-HD scheme has to allocate bandwidth resource to the MRs for communicating with the BS. For NCCG-FD, it has the worst system performance than the other two schemes. That is because this scheme can not guarantee that as users leave the old coalition to join the new one, the average system throughput will increase. In addition, we plot the system performance of the three schemes in terms of single user throughput under the same setting in Fig.~\ref{fig6}. Fig.~\ref{fig7}(a) and Fig.~\ref{fig7}(b) respectively depict the throughput of single user associated with the BS and MRs. From the two figures, we can observe that as users increasing, the three schemes all decrease. That is because when the number of users increases, the bandwidth resource that can be utilized by single user decreases, so the single user throughput decreases. However, the proposed CG-FD scheme performs the best than the other two baseline schemes in terms of single user throughput. Also, we note that the two FD schemes all perform better than the HD scheme in Fig.~\ref{fig7}(b) than Fig.~\ref{fig7}(a). The reason behind this phenomenon is that the CG-HD scheme has to allocate bandwidth resource to the MRs for communicating with the BS while the CG-FD scheme and the NCCG-FD scheme do not need.

\begin{figure}[!t]
	\begin{center}
		\includegraphics*[width=3.2in]{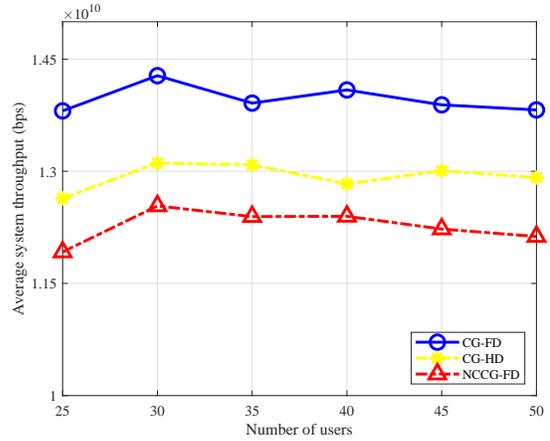}
	\end{center}
	\caption{Average system throughput of the three user association schemes under different numbers of users.}
	\label{fig6}
\end{figure}

\begin{figure}[!t]
	\centering
	\subfigure[]{
		\begin{minipage}{3.2in}
			\centering
			\includegraphics[width=3.2in]{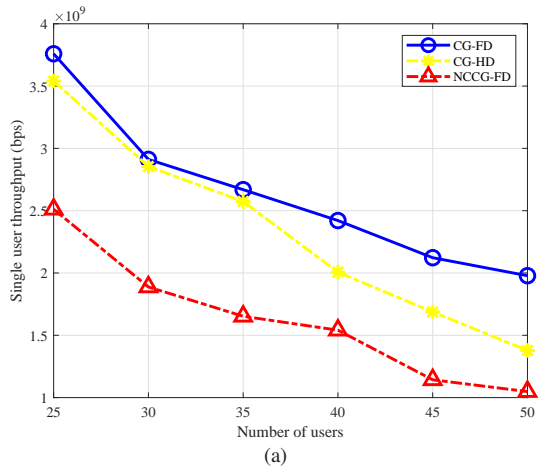}
		\end{minipage}
	}

	\subfigure[]{
		\begin{minipage}{3.2in}
			\centering
			\includegraphics[width=3.2in]{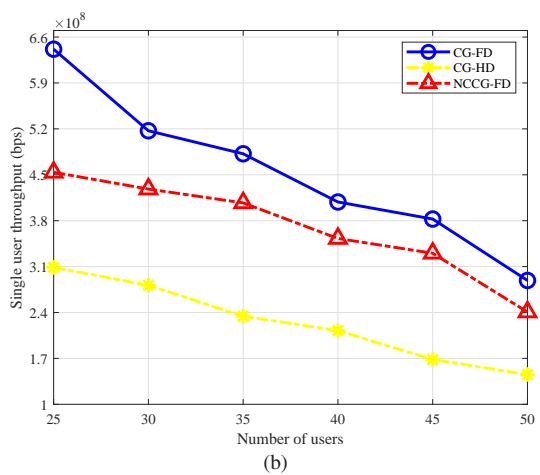}
		\end{minipage}
	}
	\caption{Single user throughput comparison of the three user association schemes under different numbers of users (a) single user associated with the BS; and (b) single user associated with MRs.}
\label{fig7}
\end{figure}

\subsubsection{Impact of $P_b$}

In Fig.~\ref{fig8}, the number of users is 40 and the number of MRs is 2. We plot the system performance of the three schemes in terms of average system throughput while the transmit power of the BS is varied from 20 dBm to 45 dBm. In Fig.~\ref{fig8}, it is obvious that as the BS transmit power increases, all the three schemes perform better. This is because when the transmit power of the BS is smaller, there will be fewer users associated with the BS and more users associated with the MRs, which leads to less bandwidth resource utilized by the users, and the transmission rate of each user will be smaller and the average system throughput decreases.

\begin{figure}[!t]
	\begin{center}
		\includegraphics*[width=3.2in]{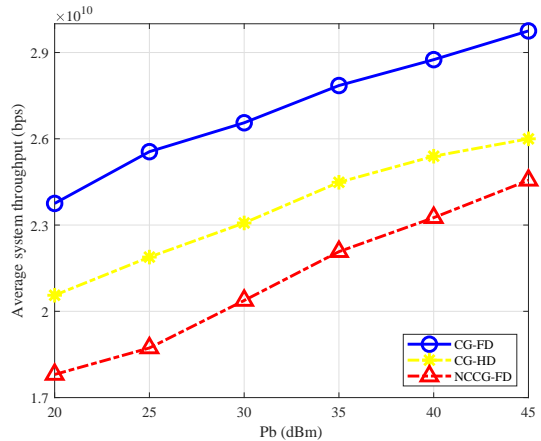}
	\end{center}
	\caption{Average system throughput of the three user association schemes under different $P_{b}$ values.}
	\label{fig8}
\end{figure}

\subsubsection{Impact of $P_n$}

In Fig.~\ref{fig9}, the parameters are set similarly to that in Fig.~\ref{fig8}. We plot the system performance of the three schemes in terms of average system throughput while the MR transmit power is changed from 5 dBm to 25 dBm. It is obvious that as the MR transmit power increasing, the average system throughputs of all the three schemes increase. This is because when the MR transmit power becomes larger, there will be more users associated with the MRs, and the burden on the BS can be alleviated. Thus, the average system throughput increase. Comparing the behaviors of the three schemes, the proposed CG-FD scheme achieves the highest system average throughput, while NCCG-FD has the worst performance.

\begin{figure}[!t]
	\begin{center}
		\includegraphics*[width=3.2in]{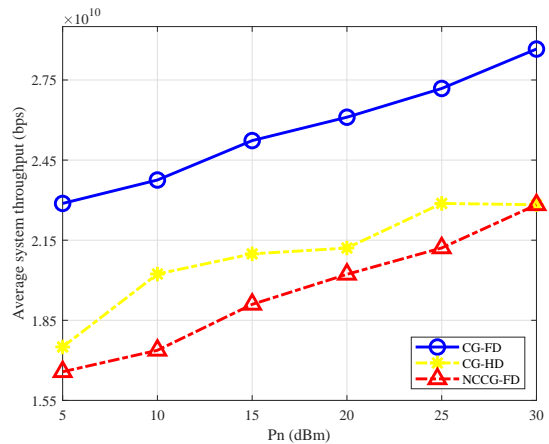}
	\end{center}
	\caption{Average system throughput of the three user association schemes under different $P_{n}$ values.}
	\label{fig9}
\end{figure}

\subsubsection{Impact of $\alpha$}

In Fig.~\ref{fig10}, we set the number of MRs to 2 and the number of users to 40, and vary the value of ${\alpha}$ from 0.1 to 0.6. We can see that with the increase of bandwidth allocation proportion of the BS, the average system throughput of all the schemes decrease gradually. Since we mainly take the users on the train into consideration, the larger the bandwidth allocation proportion of the BS, the less bandwidth resource can be utilized by users associated with MRs, and then the transmission rate of each user can be smaller. Thus, the average system throughput of the algorithms CG-FD, CG-HD and NCCG-FD will decrease. The CG-FD scheme also outperforms the two baseline schemes and NCCG-FD has the worst system performance.

\begin{figure}[!t]
	\begin{center}
		\includegraphics*[width=3.2in]{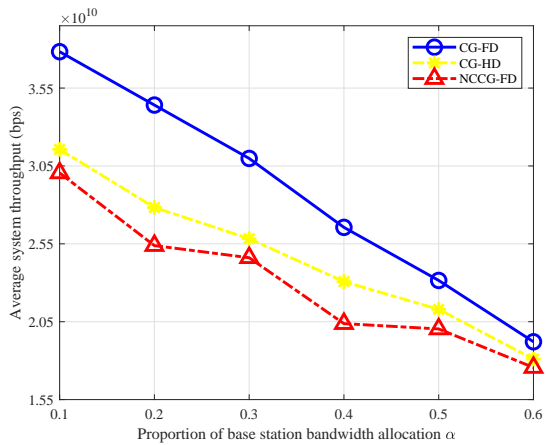}
	\end{center}
	\caption{Average system throughput of the three user association schemes under different ${\alpha}$ values.}
	\label{fig10}
\end{figure}

\subsection{Complexity}\label{S5-4}

The level of complexity is demonstrated by the number of switch operations. In Fig.~\ref{fig11}, we plot the system performance of
the three schemes in terms of the number of switch operations under the number of MRs is 2 and 4, in other words, the number of coalitions is 3 and 5, and the number of users is varied from 25 to 50. From Fig.~\ref{fig10}, it is obvious that the number of switch operations increase as the number of users increase. As more users to be associated, there are more switch operations. However, regardless of the number of users is, the number of switch operations are finite and can be inferred that the CG-FD scheme is of low complexity.

\begin{figure}[!t]
	\begin{center}
		\includegraphics*[width=3.2in]{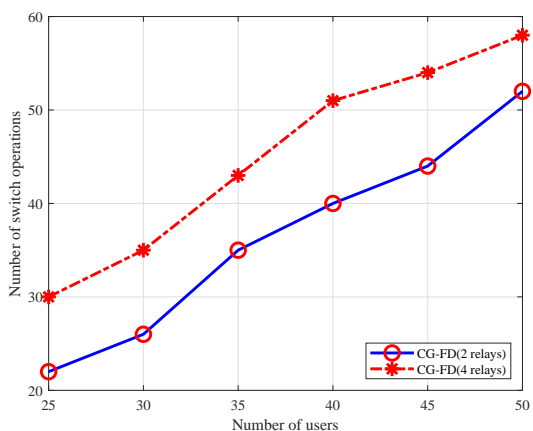}
	\end{center}
	\caption{Number of switch operations under different numbers of MRs.}
	\label{fig11}
\end{figure}

\section{Conclusions}\label{S6}

In this paper, we focused on the user association optimization problem for mmWave mobile relay systems with FD transmissions in HSR scenarios. To achieve the goal of maximizing the average system throughput, we adopted a tool called coalition game to optimize the user association problem and further proposed an algorithm based coalition game. In addition, we theoretically proved that our proposed algorithm was convergent after a limited number of switch operations and could reach Nash-stable. Moreover, lots of simulations under various system parameters were done to show the benefits of our proposed CG-FD scheme in terms of average system throughput and complexity in comparison with two baseline schemes. For future work, we will take D2D communications into consideration and jointly optimize spectrum allocation and user association.

\end{document}